\documentclass[11pt]{article}

\usepackage{amsmath,amssymb}
\usepackage{fullpage}
\usepackage{amsthm}
\usepackage{pdfpages}
\usepackage{mathrsfs}
\usepackage{comment}
\usepackage{times}

\usepackage{algorithm2e}

\usepackage[letterpaper, margin=1in]{geometry}

\theoremstyle{plain}
\newtheorem{theorem}{Theorem}
\newtheorem{lemma}[theorem]{Lemma}
\newtheorem{claim}[theorem]{Claim}

\newtheorem{corollary}[theorem]{Corollary}

\newenvironment{proofof}[1]{ {\noindent \em Proof of #1.}\/}{\hfill\qedsymbol\bigskip}

\newcommand{\remove}[1]{}
\newcommand{\suppress}[1]{}

\DeclareMathOperator{\weights}{weights}
\DeclareMathOperator{\budget}{budget}
\DeclareMathOperator{\greedy}{greedy}

\usepackage[utf8]{inputenc}
\usepackage{microtype}

\DeclareMathOperator{\MST}{MST}

\usepackage{tikz}
\usetikzlibrary{arrows.meta}

\usepackage{color}

\usepackage{graphicx} 
\usepackage{MnSymbol}

\title{Budget and Profit Approximations for Spanning Tree Interdiction}
\author{Rafail Ostrovsky\thanks{University of California, Los Angeles, {\tt rafail@cs.ucla.edu}.
Distribution Statement "A" (Approved for Public Release, Distribution Unlimited).
This material is based upon work supported by the Defense Advanced Research Projects Agency (DARPA) under Contract No. HR001123C0029. 
 Any opinions, findings and conclusions or recommendations expressed in this material are those of the author(s) and do not necessarily reflect the views of the Defense Advanced Research Projects Agency (DARPA).}\and
             Yuval Rabani\thanks{The Rachel and Selim Benin School of Computer Science and Engineering,
		The Hebrew University of Jerusalem, Jerusalem, Israel, {\tt yrabani@cs.huji.ac.il}.
		Research supported in part by ISF grants 3565-21 and 389-22, and by BSF grant 2023607.}\and
             Yoav Siman Tov\thanks{The Rachel and Selim Benin School of Computer Science and Engineering,
		The Hebrew University of Jerusalem, Jerusalem, Israel, {\tt yoav.simantov1@mail.huji.ac.il}.}}
\date{\today}

\begin{document}

\maketitle

\begin{abstract}
We give polynomial time  logarithmic approximation guarantees for the budget minimization,
as well as for the profit maximization versions of minimum spanning tree interdiction. In 
this problem, the goal is to remove some
edges of an undirected graph with edge weights and edge costs, so as to increase the weight
of a minimum spanning tree. In the budget minimization version, the goal is to minimize
the total cost of the removed edges, while achieving a desired increase $\Delta$ in the
weight of the minimum spanning tree. An alternative objective within the same framework
is to maximize the profit of interdiction, namely the increase in the weight of the minimum
spanning tree, subject to a budget constraint. There are known polynomial time $O(1)$
approximation guarantees for a similar objective (maximizing the total cost of the tree,
rather than the increase). However, the guarantee does not seem to apply to the increase
in cost. Moreover, the same techniques do not seem to apply to the budget version.

Our approximation guarantees are motivated by studying the question of minimizing
the cost of increasing the minimum spanning tree by any amount. We show that in
contrast to the budget and profit problems, this version of interdiction is polynomial
time-solvable, and we give an efficient algorithm for solving it. The solution motivates a
graph-theoretic relaxation of the NP-hard interdiction problem. The gain in minimum
spanning tree weight, as a function of the set of removed edges, is super-modular. Thus,
the budget problem is an instance of minimizing a linear function subject to a super-modular
covering constraint. We use the graph-theoretic relaxation to design and analyze a
batch greedy-based algorithm.
\end{abstract}

\thispagestyle{empty}
\newpage
\setcounter{page}{1}

\section{Introduction}

\paragraph{Problem statement and results.}
This paper deals with spanning tree interdiction. The basic setting is an undirected finite graph
with edge weights and edge costs. By removing some edges, the weight of a minimum spanning
tree can be increased, at a cost equal to the sum of costs of the removed edges. This setting
gives rise to various optimization formulations. We consider primarily the following case, which
we call the budget problem: we are given a target $\Delta$ by which to increase the weight of a
minimum spanning tree, and the optimization objective is to minimize the cost of doing so.
In the alternative profit problem, we are given a budget $B$, and the optimization objective is
to maximize the damage, namely the increase in the weight of a minimum spanning tree, subject
to the cost not exceeding $B$.

Both problems are NP-hard (see~\cite{LC93}). Our main result is a polynomial time
$O(\log n)$ approximation algorithm for the budget problem. The approximation algorithm
is motivated by considering the following special case: minimize the cost of increasing the
weight of a minimum spanning tree, by any amount. We show that in contrast with the general
budget problem, this objective can be optimized in polynomial time, and we give an efficient
algorithm for computing the optimum. We also give $O(\log n)$-approximation guarantees for
the profit problem. Finally, we investigate the question of defending against interdiction 
by adding edges to the input graph. The set of edges to choose from is given, and each edge 
is endowed with a cost of constructing it.

\paragraph{Motivation.}
The primary application of interdiction computations is to examine the sensitivity of
combinatorial optimization solutions to partial destruction of the underlying structure.
This can be used either to detect vulnerabilities in desirable structures, or to utilize
vulnerabilities to impair undesirable structures. The budget problem is perhaps more
suitable in the former setting, as its solution indicates the cost of inflicting a (dangerous)
level of damage. The profit problem is perhaps more suitable in the latter setting,
as it aims to maximize the damage inflicted using limited resources. Such problems
arise in a variety of application areas, including military planning, infrastructure protection,
law enforcement, epidemiology, etc. (see, for example the references in~\cite{Zen15}).

\paragraph{Related work.}
Previous work on spanning tree interdiction focuses exclusively on a version of the
profit problem. It approximates the total weight of the post-interdiction minimum spanning
tree, rather than the increase $\Delta$ in the weight of the tree as per the above definition
of the profit problem. Notice that if the resulting tree has total weight which is $C$ times
the weight of the initial tree, then approximating the total weight by any factor at least $C$
means the algorithm could end up doing nothing. In contrast, approximation of the profit
problem guarantees actual interdiction even if $\Delta$ is very small compared with the
weight of the initial tree. Note that in the case of the budget problem there is no qualitative
difference between specifying the target total weight and specifying the target increase in
weight.

The case of uniform cost was first considered in~\cite{FS96} who gave a poly-time
$O(\log B)$ approximation algorithm for the (total tree weight version of the) profit problem,
where $B$ is the budget (i.e., the number of edges the algorithm is allowed to remove).
They showed that the uniform cost problem is NP-hard (previously, it was known that the
problem with arbitrary costs is NP-hard; see~\cite{LC93} and the references in~\cite{FS96}).
The same problem was also discussed in~\cite{BTV11} and the references therein, where
algorithms running in time that is exponential in the budget $B$ were considered. Later,
constant factor approximation algorithms for the problem, without the cost-uniformity constraint,
were found~\cite{Zen15,LS17}. The latter paper gave a $4$-approximation guarantee.
The upper bound that was used in both papers cannot be used to get an approximation
better than $3$~\cite{LS17}. In~\cite{GS14} it was shown that the problem is fixed parameter-tractable
(parametrized by the budget $B$), but the budget problem is  $W[1]$-hard (parametrized
by the weight of the resulting tree).

We briefly review the constant factor approximation guarantees for total minimum spanning
tree weight in~\cite{Zen15,LS17}. Both papers use the following framework.
($i$) Let $w_1\le w_2\le \cdots\le w_k$ be the sorted list of distinct edge weights, and let
$G_1,G_2,\dots,G_k$ denote the subgraphs of the input graph $G$, where the edges of
$G_i$ are all the edges of $G$ of weight at most $w_i$. Then, the objective function at a
set $F$ of removed edges can be expressed as a function of the number of connected
components of $G_i\setminus F$ and $w_i - w_{i-1}$, for all $i$. This implies, in particular,
that the objective function is super-modular.
($ii$) Maximizing an unconstrained super-modular function is a polynomial time-computable
problem. Hence, the Lagrangian relaxation of the linear budget constraint can be computed
efficiently for any fixed setting of the Lagrange multiplier. Binary search can be used to find
a good multiplier. The usual impediment of this approach shows up here as well. The search
resolves the problem only if the solution spends exactly the upper bound on the cost. However,
it may end up producing two solutions for (essentially) the same Lagrange multiplier, one below
budget and one above budget. Those combine to form a bi-point {\em fractional} solution.
($iii$) How to extract a good integral solution from this bi-point solution is where the papers
diverge. The tighter approximation of~\cite{LS17} reduces, approximately, the problem of
extracting a good solution to the problem of tree knapsack, then uses a greedy method to
approximate the latter problem. The former result of~\cite{Zen15} used a more complicated
argument, but also a greedy approach. We give an example (in Section~\ref{sec:profit}) that
these algorithms do not perform well in terms of approximating the increase $\Delta$ in spanning
tree weight.

Super-modularity carries over to the objective function we use here, namely the increase
in the weight of a minimum spanning tree, as the difference between the functions is a
constant (the weight of the spanning tree before interdiction). Thus, as the above discussion
hints to, the profit problem is a special case of the problem of maximizing a monotonically
non-decreasing and non-negative super-modular function subject to a linear packing constraint
(a.k.a. a knapsack constraint). Similarly, the budget problem is a special case of minimizing a
non-negative linear function subject to a super-modular covering constraint (i.e., a lower bound
on a non-decreasing and non-negative super-modular function).

Similar settings are prevalent in combinatorial optimization. A generic problem of this
flavor is set cover, which is a special case of minimizing a non-negative linear function
subject to a monotonically non-decreasing and non-negative {\em sub-modular} covering
constraint. The related maximum coverage problem is a special case of maximizing a
monotonically non-decreasing and non-negative sub-modular function, subject to a
cardinality constraint (which is, of course, a special case of a knapsack constraint). More
broadly, many problems that arise in unsupervised machine learning are of this flavor.
For example, $k$-means and $k$-median clustering are special cases of minimizing
a monotonically non-increasing and non-negative sub-modular function subject to a
cardinality constraint. Several such formulations received general treatment; see for
example~\cite{Svi04,SF08,FNS11,VCZ11,IB13,SVW15,LS17b,AFLLMR21} and the
references therein. Obviously, an optimization problem has equivalent representations
derived by transformations between super-modularity and sub-modularity, maximization
and minimization, and/or covering and packing (by defining the function over the complement
set, or negating). These transformations reverse monotonicity, and moreover may not
preserve approximation bounds.

The particular combination of super-modular maximization subject to a knapsack constraint
is known as the super-modular knapsack problem, introduced in~\cite{GS89}. In general, it
is hard to approximate within any factor (given query access to a monotonically non-decreasing
objective function subject to a cardinality constraint); see the example in~\cite{Usu16}. The
case of a {\em symmetric} (therefore, non-monotone) super-modular function can be solved
exactly in polynomial time~\cite{GS10}. We are not aware of any relevant work on the problem
of approximating the minimum of a non-negative linear objective, subject to a super-modular
covering constraint. The convex hull of the indicator vectors that satisfy a generic super-modular
covering constraint is investigated in~\cite{AB15}.

Finally, we mention that spanning tree interdiction is one problem in a large repertoire
of interdiction problems, including in particular interdiction of shortest path, assignment
and matching problems, network flow problems, linear programs, etc. Some representative
papers
include~\cite{KBBEGRZ08,Zen08,BEHKKPSS10,BTV13,DG13,HMMPRS17}
(this list is far from being comprehensive).

\paragraph{Our techniques.}
Our results rely on
the notion of a partial cut, which is the set of edges that cross a cut with weight
below a given threshold weight. An optimal solution minimizing the cost of
increasing the weight of a minimum spanning tree by any amount is a single
partial cut. We show that such a cut can be computed efficiently by enumerating
over a polynomial time computable collection of candidate partial cuts. In order
to derive the approximation guarantees for the general budget problem, we apply
a batch greedy approach. We repeatedly compute the collection of candidate partial
cuts and choose a cut with the best gain per cost ratio. The proof of approximation
guarantee relies on an approximate characterization of an optimal solution by a
collection of partial cuts. We further show how to speed up the computation by
using in all iterations the collection computed for the input graph, rather than
recomputing a new collection in each iteration. A similar approach gives the logarithmic 
approximation for the profit problem, in a manner parallel to the greedy approximation 
for knapsack (i.e. use either the maximum greedy solution that does not exceed the 
budget, or the best single partial cut).

\paragraph{Organization.}
In Section~\ref{sec:prelim} we present some useful definitions and claims,
including a self-contained (and different) proof of super-modularity of the
gain in spanning tree weight function. In Section~\ref{sec:eps-inc} we give
a polynomial time algorithm for minimizing the cost of increasing the minimum
spanning tree weight by any amount. This algorithm motivates our approximation
algorithm for the budget problem. In Section~\ref{sec:relax} we present our
graph-theoretic relaxation. In Section~\ref{sec:budget} we present our main
result---an approximation algorithm for the budget problem. In Section~\ref{sec:profit}
we give an approximation algorithm for the profit problem, and discuss the
shortcoming of previous work to achieve this objective. Finally, in 
Section~\ref{sec:def} we remark on defense against spanning tree interdiction.

\section{Preliminaries}\label{sec:prelim}

In this section we present some general definitions and useful lemmas.

Let $G = (V, E)$ be a weighted undirected graph such that every edge $e$ has a non-negative weight
$w:E\rightarrow \mathbb{R}^+ \cup \{0\}$ and a positive removal cost $c:E\rightarrow \mathbb{R}^+$.
Let $\MST(G)$ denote the weight of a minimum spanning tree of $G$. We'll use the convention that
if $G$ is disconnected, then $\MST(G) = \infty$. Also, for a set of edges $F\subset E$, we denote
$c(F) = \sum_{f\in F} c(f)$.
Given a budget $B$, the {\em spanning tree Interdiction problem} is to find a set of edges $F\subset E$
satisfying $c(F) \leq B$ and maximizing  $\MST(G\setminus F)$, where $G\setminus F$ denotes the
graph $(V, E \setminus F)$.
The {\em profit} $p_G(F)$ of a solution $F$ to the spanning tree interdiction problem is defined to be
the increase $p_G(F) = \MST(G\setminus F) - \MST(G)$ in the weight of the minimum spanning tree.
The {\em profit to cost ratio} of a set of edges $F\subset E$ is defined to be $r_{G}(F) = \frac{p_{G}(F)}{c(F)}$.
For the empty set ($c(F)=0$), we define $r_{G}(\emptyset) = 0$.

Consider a weighted graph $G = (V,E)$ as above. Given a set of nodes $S$,
$\emptyset\ne S\subsetneq V$,
The {\em complete cut} $C = C_G(S)$ defined by $S$ is the set of edges
\[C = C_G(S) = \{ e \in E:\ |e\cap S|=1 \}.\]
We say that the edges in $C$ {\em cross} the cut $C = C_G(S)$.
Given $S$ and $W\in\mathbb{R}^+$, the {\em partial cut} $C = C_G(S,W)$ is the set of edges
\[C = C_G(S,W) = \{ e \in C_G(S):\ w(e) < W \}.\]

Consider a connected graph $G$, let $T$ be a spanning tree of $G$, and let $e \in T$. We denote by
$C_{T,e}$ the cut in $G$ that satisfies $C_{T,e} \cap T = e$.
\\
We show the following lemmas that will be useful later (the proofs are in Appendix~\ref{appendix:A}).

\begin{lemma}~\label{MST_DELETED_EDGES_CLAIM}
Consider a minimum spanning tree $T$ of a connected graph $G$. Let $F\subset E$
such that $G' = G \setminus F$ is connected. Then, there exists a minimum spanning
tree $T'$ for $G'$ that includes all the edges in $T \setminus F$.
\end{lemma}

\begin{lemma}\label{PROFIT_OF_CUT}
Let $G = (V,E)$ be a $w$-weighted graph and let $C = C_G(S,W) \subseteq E$ be a partial cut in $G$.
Let $e = (u,v) \in E$ be an edge that crosses $C_G(S)$. Then,
\[p_G(C) \geq W - w(e).\]
\end{lemma}

The following lemma reproves a claim from~\cite{Zen15}.
\begin{lemma}[super-modularity of the profit function]\label{SUPERMODULARITY}
Let graph $G = (V,E)$ be a $w$-weighted graph, let $B \subset E$ be set of edges, and let
$e \in E\setminus B$ be an edge. If $G\setminus B$ is connected, then
\[p_{G \setminus B}(e) \ge p_G(e).\]
\end{lemma}

\begin{corollary}~\label{PROFIT_PRESERVATION}
Let $G=(V,E)$ be a weighted graph, and let $A,B \subset E$ be disjoint sets of edges
($A \cap B = \emptyset$). Then
$p_{G \setminus B}(A) \ge p_G(A)$.
\end{corollary}

\section{An Algorithm for \boldmath$\varepsilon$-Increase}\label{sec:eps-inc}

In this section we design and analyze a polynomial time algorithm for computing
the minimum cost interdiction to increase the weight of a minimum spanning tree
(by any amount). I.e., given a graph $G=(V,E)$ with edge weights $w$ and edge
costs $c$, we want to find a set of edges $F\subset E$ for which
$\MST(G\setminus F) > \MST(G)$, minimizing $c(F)$. This algorithm motivates
our approximation algorithm for the budget problem given in Section~\ref{sec:budget}
(and its derivative for the profit problem in Section~\ref{sec:profit}).

The algorithm is defined as follows. Compute a minimum spanning tree $T$ of
$G$. Enumerate over all the edges $e\in T$. Given an edge $e$, contract all the
edges of weight $< w(e)$.
Remove all edges of length $> w(e)$. Find a minimum (with respect to edge
cost) $u$-$v$ cut in the resulting graph, where $e = \{u,v\}$. The output of
the algorithm $F_{\min}$ is the minimum cost cut generated, among all choices
of $e \in T$.\\
\\
The following two claims imply that the output of the algorithm is valid and optimal (the proofs are in Appendix~\ref{appendix:A}).

\begin{claim}\label{cl: alg cost}
$c(F_{\min})\le c(F^*)$, where $F^*$ is an optimal solution.
\end{claim}

\begin{claim}\label{cl: alg len}
$\MST(G \setminus F_{\min}) > \MST(G)$.
\end{claim}

Let $\tau(n,m)$ denote the time to compute a minimum $s$-$t$ cut in a graph
with $n$ nodes and $m$ edges.
\begin{corollary}
The algorithm finds an optimal solution in time $O(n \cdot \tau(|V|,|E|))$.
\end{corollary}

\begin{proof}
Recall that $F_{\min}$ is the solution that the algorithm computes. By Claim~\ref{cl: alg cost},
$c(F_{\min}) \le c(F^*)$, and by Claim~\ref{cl: alg len} we have $\MST(G \setminus F_{\min}) > \MST(G)$,
so $F_{\min}$ is an optimal solution. The algorithm iterates over the $|V|-1$ edges of $T$, and
for each edge calculates a minimum cut, hence the time complexity.
\end{proof}

\section{Relaxed Specification of the Optimum}\label{sec:relax}

In this section, we define a relaxation to the optimal solution for the budget minimization problem,
based on a carefully constructed collection of partial cuts. We will then use this relaxation to
analyze our approximation algorithms.

Given a solution $F$ with cost $B=c(F)$ and profit $\Delta$, we construct a sequence of cuts $C_1, C_2, \ldots, C_{t-1}$
that satisfy the following properties regarding their cost and profit.

\begin{theorem}~\label{CUTS_THEOREM}
Let $G = (V,E)$ be an undirected graph, and let $n = |V|$. Also, let
$F \subseteq E$ be a set of edges of cost $B = c(F)$ and profit $\Delta$.
Then, there exists a sequence of partial cuts $C_1, C_2, \ldots, C_{t-1}$ such that
\[\sum_{i=1}^{t-1} c(C_i) \leq 2B\cdot \log n,\]
and
\[\sum_{i=1}^{t-1} p_G(C_i) \geq \Delta.\]
\end{theorem}

\paragraph{Constructing the sequence of the cuts.}
Let $G = (V,E)$ be a graph, and let $T$ be a minimum spanning tree in $G$. Let $F \subset E$
be a set of edges with cost $B$ and profit $\Delta$. Denote $G \setminus F$ by $G'$. The solution
$F$ removes some edges of $G$, including $t-1$ edges $e_1, e_2, \ldots, e_{t-1}$ of $T$. Notice
that $t \leq n-1$, and if all removal costs are $1$, also $t \leq B$. Removing those edges splits $T$
into $t$ connected components $A_1, A_2, \ldots, A_{t}$. We emphasize that $A_i$ denotes the
node set of the $i$-th connected component. Therefore, this is a partition of the nodes of $G$.
Let $T'$ be a minimum spanning tree of $G'$. By Lemma~\ref{MST_DELETED_EDGES_CLAIM}, we can
choose $T'$ which uses the same edges as $T$ inside the connected components $A_1, A_2, \ldots, A_{t}$,
and reconnects these components using $t-1$ new edges to replace the removed edges $e_1, e_2, \ldots, e_{t-1}$.

In the following construction we consider the connected components graph $G_{cc}' = (V_{cc},E_{cc}')$, where
$V_{cc} = \{A_1, A_2, \ldots, A_{t}\}$,
and $E_{cc}'$ includes an edge between $A_i$ and $A_j$ for
every pair of vertices $u\in A_i$ and $v\in A_j$ such that $\{u,v\}\in E'$ (with the same weight as the
corresponding edge in $G'$). Notice that $G_{cc}'$ may have many parallel edges, and every edge
of $G_{cc}'$ corresponds to an edge of $G'$. To avoid notational clutter, we will use the same notation
for an edge of $G_{cc}'$ and for the corresponding edge of $G'$. Also, we will use the same notation
for a vertex of $G_{cc}'$ and for the corresponding set of vertices of $G'$ (which is the same as the
vertices of $G$), as well as for a set of vertices
of $G_{cc}'$ and for the union of the corresponding sets of vertices of $G'$. The interpretation will be clear
from the context. The idea behind defining $G_{cc}'$ is to hide the edges that $T$ and $T'$ share inside
the connected components, so that the cuts we construct can't delete them. We use $G_{cc}'$ only in
order to construct the sequence of cuts (these are cuts in $G$).
Let $T_{cc}'$ be a minimum spanning tree of $G_{cc}'$. Note that $T_{cc}'$ has $t-1$ edges, which
we denote by $e_1', e_2', \ldots, e_{t-1}'$. These edges correspond exactly to the new edges of
a minimum spanning tree $T'$ of $G'$ that replace the edges $e_1, e_2, \ldots, e_{t-1}\in T\cap F$.
For the construction, a strict total order on the weights of the edges of $T_{cc}'$ is needed. We
index these edges in non-decreasing order, breaking ties arbitrarily.

\paragraph{Two alternatives.}
For each edge $e_i'$, we first define two alternatives for a partial cut, $C_i^R$ or $C_i^L$, and later
we choose only one of them. This choice is repeated for every edge in $T_{cc}'$ to get the desired
collection of $t-1$ cuts.
Consider $e_i'$ to be an edge in $T_{cc}'$ of weight $w(e_i')$. Delete from $T_{cc}'$ all the edges
$e_j'$, $j\ge i$. Consider the connected components of the resulting forest. Notice that the edge
$e_i'$ must connect between two such components $L_i$ and $R_i$ (recall that these denote
both sets of vertices of $G_{cc}'$ and the corresponding unions of sets of vertices of $G$).
Define the following two partial cuts: $C_i^L = C_G(L_i, w(e_i'))$ and $C_i^R = C_G(R_i, w(e_i'))$.
We will denote by $X_i$ the choice we make between $L_i$ and $R_i$, which we refer to as {\em the
small side} of the cut in step $i$. Also, we put $C_i = C_G(X_i, w(e_i'))$, the cut chosen in step $i$.

\usetikzlibrary{shapes.geometric, intersections}
\begin{figure}
\centering
\tikzset{mynode/.style={draw, very thick, circle, minimum size=1cm}, myarrow/.style={very thick}}

\begin{tikzpicture}
\node[mynode](n5) at (2,-1.5){$b_5$};
\node[mynode](n6) at (60:3){$b_6$};
\node[mynode](n3) at (132:3.2){$b_3$};
\node[mynode](n2) at (204:4){$b_2$};
\node[mynode](n4) at (-15:6){$b_4$};
\node[mynode](n7) at (25:6){$b_7$};
\node[mynode](n1) at (160:6){$b_1$};
\node[mynode](n8) at (16:9){$b_8$};

\draw[myarrow](n1)--(n3) node[midway,above left] {10};
\draw[myarrow](n1)--(n2) node[midway,right] {7};
\draw[myarrow](n3)--(n5) node[midway,below left] {12};
\draw[myarrow](n3)--(n6) node[midway,above] {3};
\draw[myarrow](n6)--(n7) node[midway,above] {5};
\draw[myarrow](n6)--(n7) node[midway,below] {$e_i'$};

\draw[myarrow](n5)--(n4) node[midway, above] {15};
\draw[myarrow](n7)--(n8) node[midway, above] {2};

\node[ellipse, draw, minimum width = 5cm, minimum height = 3.1cm, dashed] (e5) at (7.1,2.5){};
\node[ellipse, draw, minimum width = 6cm, minimum height = 3.1cm, dashed] (e5) at (-0.25,2.5){};

\node (text) at (7.2,3.7) {$C_i^R$};
\node (text) at (-0.2,3.7) {$C_i^L$};

\end{tikzpicture}
\caption{The two options of some edge $e_i'$ over the MST of the connected components graph.}~\label{CREATE_CUT_FIGURE1}
\end{figure}
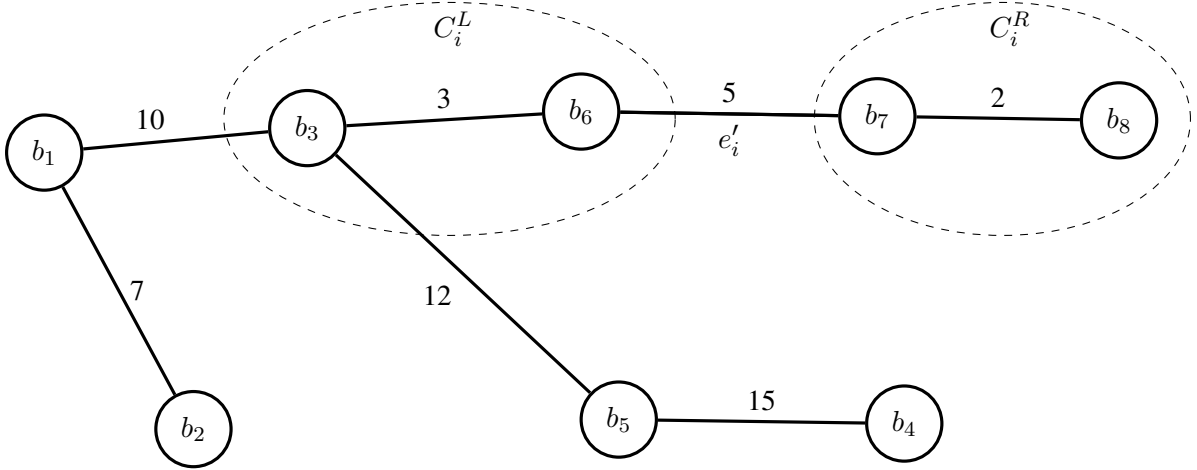

\paragraph{\boldmath Choosing $X_i$.} The goal is that any $A_j \in V_{cc}$ will not be contained in
``too many" $X_i$-s. We count for each vertex $A_j \in V_{cc}$ how many times it was chosen to be in
the small side. Denote this number as $k(A_j)$, and for a set of vertices $S \subseteq V_{cc}$ denote
$k(S) = max_{A_j \in S}(k(A_j))$.

We choose cuts in ascending order of $i$. If $k(L_i) \le k(R_i)$, we choose $X_i = L_i$, and otherwise we
choose $X_i = R_i$. After choosing a new $X_i$, we increase by $1$ the counter $k(A_j)$ for every
vertex $A_j \in X_i$, then proceed to choosing the next cut.

The following lemma shows that the $X_i-s$ form a laminar set system over the vertices.
\begin{lemma}\label{lm: laminar}
Let $i > j$. Then, either $X_i\cap X_j = \emptyset$, or $X_i\supset X_j$.
\end{lemma}

\begin{proof}
Assume that $X_j \cap X_i \ne \emptyset$. Let $A \in X_j \cap X_i$ be a common vertex. Consider a vertex $A' \in X_j$. 
Clearly, there exists a path $P$ connecting $A$ and $A'$ in the tree $T_{cc}'$, such that every edge in this path precedes $e_j'$ in the non-decreasing order; otherwise, $A$ and $A'$ would not be in the same connected component of the tree $T_{cc}'$ after deleting all the edges preceding $e_j'$.
Because $e_i'$ comes after $e_j'$, it also holds that both $A$ and
$A'$ are in the same component of the tree $T_{cc}'$ after deleting all the edges with index at least $i$.
This implies that if $A \in X_i$, then also $A' \in X_i$, and therefore $X_j \subset X_i$.
\end{proof}

\usetikzlibrary{shapes.geometric, intersections}
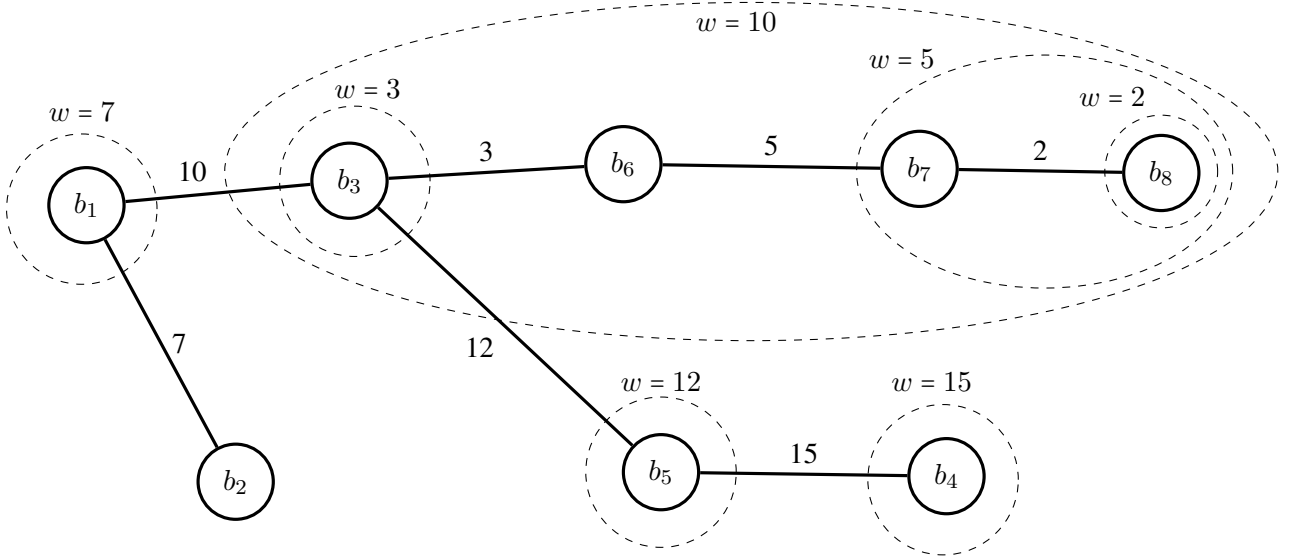
\begin{figure}
\centering
\tikzset{mynode/.style={draw, very thick, circle, minimum size=1cm}, myarrow/.style={very thick}}

\begin{tikzpicture}
\node[mynode](n5) at (2,-1.5){$b_5$};
\node[mynode](n6) at (60:3){$b_6$};
\node[mynode](n3) at (132:3.2){$b_3$};
\node[mynode](n2) at (204:4){$b_2$};
\node[mynode](n4) at (-15:6){$b_4$};
\node[mynode](n7) at (25:6){$b_7$};
\node[mynode](n1) at (160:6){$b_1$};
\node[mynode](n8) at (16:9){$b_8$};

\draw[myarrow](n1)--(n3) node[midway,above left] {10};
\draw[myarrow](n1)--(n2) node[midway,right] {7};
\draw[myarrow](n3)--(n5) node[midway,below left] {12};
\draw[myarrow](n3)--(n6) node[midway,above] {3};
\draw[myarrow](n6)--(n7) node[midway,above] {5};
\draw[myarrow](n5)--(n4) node[midway, above] {15};
\draw[myarrow](n7)--(n8) node[midway, above] {2};

\node[ellipse, draw, minimum width = 14cm, minimum height = 4.5cm, dashed] (e1) at (3.2,2.5){};
\node[ellipse, draw, minimum width = 2cm, minimum height = 2cm, dashed] (e2) at (2,-1.5){};
\node[ellipse, draw, minimum width = 2cm, minimum height = 2cm, dashed] (e3) at (5.75,-1.6){};
\node[ellipse, draw, minimum width = 2cm, minimum height = 2cm, dashed] (e4) at (-5.7,2){};

\node[ellipse, draw, minimum width = 2cm, minimum height = 2cm, dashed] (e6) at (-2.07,2.37){};
\node[ellipse, draw, minimum width = 1.5cm, minimum height = 1.5cm, dashed] (e6) at (8.65,2.5){};
\node[ellipse, draw, minimum width = 5cm, minimum height = 3.1cm, dashed] (e5) at (7.1,2.5){};

\node (text) at (3,4.5) {$w=10$};
\node (text) at (2, -0.3) {$w=12$};
\node (text) at (5.6, -0.3) {$w=15$};
\node (text) at (-5.7,3.3) {$w=7$};
\node (text) at (5.2,4) {$w=5$};
\node (text) at (-1.9,3.6) {$w=3$};
\node (text) at (8,3.5) {$w=2$};

\end{tikzpicture}
\caption{Example for cuts created according to MST of connected components graph.}~\label{CREATE_CUT_FIGURE}
\end{figure}

We now show that the first claim of Theorem~\ref{CUTS_THEOREM} holds.
\begin{lemma}\label{SUMS_CLAIM}
The sum of the costs of the cuts $C_1, C_2, \ldots, C_{t-1}$ is
\[\sum_{i=1}^{t-1} c(C_i) \leq 2B\cdot \log t \leq 2B\cdot \log n.\]
\end{lemma}

The proof of Lemma~\ref{SUMS_CLAIM} relies on the following claims.
\begin{claim}\label{cl: cut in G}
For every $i=1,2,\dots, t$, if $e \in C_i$ then $e \in F$.
\end{claim}

\begin{proof}
Let’s assume for contradiction that there exists an edge $e \in C_i$, but $e \notin F$.
By construction, $w(e) < w(e_i')$. Recall that $C_i$ is the set of edges with exactly
one endpoint in a connected component $X_i$ of $T_{cc}'$ after removing edges of
weight $\ge w(e_i')$. In particular, $e = \{u,v\}$, where $u\in A_i\in X_i$ and
$v\in A_j\not\in X_i$. Consider the path in $T_{cc}'$ between $A_i$ and $A_j$.
There must be an edge $e'$ of weight $w(e') \geq w(e_i')$ along this path, otherwise
both $u$ and $v$ would be on the same side of the cut. We assumed that $e\not\in F$,
hence $e \in E_{cc}'$. But if $w(e) < w(e_i') \leq w(e')$ and $e \in E_{cc}'$, then replacing
$e'$ with $e$ in $T_{cc}'$ creates a spanning tree of $G_{cc}'$ which is lighter than $T_{cc}'$,
a contradiction.
\end{proof}

Next we show that no edge gets deleted by more than $2 \cdot \log t$ cuts.
\begin{claim}\label{cl: num cuts}
For any edge $e \in E$, $e$ crosses no more than $2 \cdot \log t$ cuts $C_1,C_2,\dots,C_{t-1}$.
\end{claim}

\begin{proof}
Let $\{A_1, A_2, \ldots A_{t}\}$ be the vertex set of the graph $G_{cc}'$, and consider the final counts
$k(A_1)$, $k(A_2)$, $\ldots$, $k(A_t)$ for the vertices, respectively, after the construction of the cuts
$C_1, C_2, \ldots C_{t-1}$ as described above. We show that for every $1 \leq i \leq t$, it holds that
$k(A_i) \leq \log(t)$. This means that every vertex can't be in the small side of a cut more than $\log(t)$
times, and therefore an edge can't cross more than $2 \cdot \log(t)$ cuts.

We first show that for every $A \in L_i \cup R_i$ it holds that at the end of step $i$, $k(A) \leq \log(|L_i|+|R_i|)$
(viewing $L_i,R_i$ as sets of vertices in $G_{cc}'$). The proof is by induction on $i$.

{\em Base case:} Notice that $|L_1|,|R_1|\ge 1$. Therefore, $\log(|L_i|+|R_i|) = 1$.
As only one step was executed, for every vertex $A$, we have that $k(A) \leq 1$, as required.

{\em Inductive step:} Assume the claim is true for every $j < i$. Let $s = |R_i| + |L_i|$. It must hold
that either $|L_i| \leq \frac{s}{2}$ or $|R_i| \leq \frac{s}{2}$. Assume without loss of generality that
$|R_i| \leq \frac{s}{2}$. Consider the largest value of $k(A)$ for any $A \in R_i$ at the end of step
$i-1$. For this $A$, let $X_j$ be a small side that contains $A$, for the largest such $j < i$.
As $w(e_j')\le w(e_i')$, it must be that $R_j,L_j \subset R_i$. This is true because the edges of
weight below $w(e_i')$ include the edges of weight at most $w(e_j')$, so $R_j,L_j$ are in the same
component $R_i$ (as $A \in R_i\cap X_j$). Because $R_j\cap L_j = \emptyset$ it holds that
$|R_j| + |L_j| \leq |R_i| \leq \frac{s}{2}$. By the induction hypothesis, at the end of step $j$, we
have that $k(A) \leq \log(\frac{s}{2}) = \log(s) - 1$. Therefore, this is true also at the end of step
$i-1$, because $k(A)$ did not change after step $j$ and before step $i$. If $X_i = R_i$, then after
step $i$, we get that $k(A)$ increases by $1$ to $\log(s)$, as claimed. If $X_i = L_i$, then $k(A)$
does not change in step $i$, so it holds that $k(A)\leq \log(s) - 1 <  \log(s)$.

Finally, consider $A'\in L_i$. By the same argument as for $R_i$, at the end of step $i-1$ we
have that $k(A')\leq \log(|L_i|)\leq \log(s-1) < \log(s)$. If $X_i\neq L_i$, then this holds after
step $i$. Otherwise, by the choice of $X_i$ it must hold that before step $i$,
$k(A')\leq k(A)\leq \log(s) - 1$. Therefore, after step $i$, $k(A')\leq \log (s)$.
\end{proof}

\begin{proofof}{Lemma~\ref{SUMS_CLAIM}}
By Claim~\ref{cl: cut in G}, any edge included in at least one of the partial cuts is included in $F$.

By Claim~\ref{cl: num cuts}, no edge crosses more than $2 \cdot \log(t)$ partial cuts. Because the
cuts use only edges from $F$ and use each edge up to $2 \cdot \log(t)$ times, the sum of their costs
is no more than $2B \cdot \log(t) \leq 2B \cdot \log n$.
\end{proofof}

It remains to show the second claim of Theorem~\ref{CUTS_THEOREM}.
\begin{lemma}\label{PROFIT_CLAIM}
For the cuts $C_1, C_2, \ldots, C_{t-1}$ it holds that
\[\sum_{i=1}^{t-1} p_G(C_i) \geq \Delta\]
\end{lemma}

The first step of the proof is to show a perfect matching between the edges that
were removed from $T$ (a minimum spanning tree of $G$) and the new edges that
replaced them in $T'$ (a minimum spanning tree of $G \setminus F$). We will use this
matching to argue that for every matched pair there is a cut with a profit of at least the
difference of weights between the edges, and this will be used to lower bound the
total profit. To show the existence of a perfect matching, we employ Hall's condition.
We begin with the following claim.
\begin{claim}\label{cl: typical vertex}
Let $C_1, C_2, \ldots, C_{t-1}$ be the cuts induced by $T$ and $F$. Then, for every
cut $C_i$ there exists a vertex $v_i \in V_{cc}$ such that $v_i \in X_i$, but $v_i \notin X_j$
for all $j < i$. Moreover, there exists a vertex $u \in V_{cc}$ such that $u \notin X_j$
for all $j \in [t-1]$.
\end{claim}

We refer to $v_i$ as the {\em typical vertex} of $C_i$, and to $u$ as
the {\em typical vertex} of the graph; see Figure~\ref{CREATE_CUT_FIGURE};
$b_6$ is the typical vertex of the cut with the weight of $w=10$, and $b_2$ is the
typical vertex of the graph.
\bigskip

\begin{proofof}{Claim~\ref{cl: typical vertex}}
Consider $X_i$. Every edge $e \in T_{cc}'$ with both endpoints in $X_i$ has index $<i$ and every edge $e \in T_{cc}'$ that connects between $X_i$ and
$V_{cc} \setminus X_i$ must have index $\geq i$. 
Clearly, if $X_j\cap X_i=\emptyset$,
then any choice of $v_i\in X_i$ will satisfy $v_i\not\in X_j$. By Lemma~\ref{lm: laminar},
all other $j < i$ satisfy $X_j\subsetneq X_i$.

It is sufficient to consider all $j$ such that
$X_j\subsetneq X_i$ and $\not\exists j'$ such that $X_j\subsetneq X_{j'}\subsetneq X_i$.
Let $j_1 < j_2 < \cdots < j_l$ be an enumeration of these indices. Notice that
$|X_{j_1}\cap e_{j_1}'| = 1$. Moreover, for $r > 1$, $X_{j_r}\cap e_{j_1}' = \emptyset$.
This is because $e_{j_1}'$ and all the edges with two endpoints in $X_{j_1}$ are not
deleted when constructing $X_{j_r}$. Therefore, if $X_{j_r}\cap e_{j_1}' \neq \emptyset$,
then $X_{j_r}\supset X_{j_1}$, in contradiction to the definition of $j_1$. However,
$X_i\supset X_{j_r}$ for all $r$, and $e_{j_1}'$ is not deleted when constructing
$X_i$, hence both endpoints of $e_{j_1}'$ are contained in $X_i$. The endpoint
not contained in $X_{j_1}$ can be used as the typical vertex $v_i$ of $C_i$.

The same argument, applied with $V_{cc}$ replacing $X_i$, proves the existence
of the typical vertex $u$ of $G$.
\end{proofof}

Next we prove the following claim.
\begin{claim}\label{K_EDGES_EXIST}
For every $k\in\{1,2,\dots,t-1\}$ and for every $A \subseteq \{C_1, C_2, \ldots, C_{t-1} \}$
of cardinality $|A| = k$, there exist at least $k$ distinct edges
$e_1, \ldots e_k\in F \cap T$ that cross at least one of the cuts in $A$.
\end{claim}

\begin{proof}
The edges in $F \cap T$ by definition form a spanning tree on $V_{cc}$ equipped with
the edges of the original graph $G$.

Given a set $A$ of $k$ cuts, we say that two vertices $A_i,A_j \in V_{cc}$ are {\em in the same area}
if and only if they are in the same side of any cut in $A$. Notice that the partition into areas is an
equivalence relation as it is reflexive, symmetric and transitive. We claim that there are at least
$k+1$ different areas, because there are at least $k+1$ vertices for which every pair of them is not
in the same area (separated with at least one cut). Each of the cuts in $A$ has a typical vertex. Any
pair of two typical vertices $v_i,v_j$, $i\ne j$ cannot be in the same area. If the two cuts
are disjoint ($X_i \cap X_j = \emptyset$) then $v_i\in X_i$ but $v_j\not\in X_i$. Otherwise,
one contains the other ($X_i \subset X_j$), but $v_i\in X_i$ whereas $v_j\in X_j\setminus X_i$.
Moreover, the typical vertex of the graph is not in the same area as any of the other typical
vertices (because it not in any $X_i$). To cap, in total there are at least $k+1$ areas, and any
edge between two different areas must cross at least one cut in $A$. The spanning tree $F \cap T$
has to connect, in particular, all the vertices in the different areas, so it must have at least $k$
edges that connect vertices in two different areas.
\end{proof}

\begin{claim}\label{cl: profit claim}
There exists a permutation $\pi$ on $\{1,2,\dots,t-1\}$ such that
$e_{\pi(i)}\in C_i$ and
$$
\sum_{i=1}^{t-1} \left(w(e_i') - w(e_{\pi(i)})\right) = \Delta.
$$
\end{claim}

\begin{proof}
By Claim~\ref{K_EDGES_EXIST} and Hall's marriage theorem we conclude that
there exists a perfect matching between the $t-1$ cuts and the $t-1$ edges of
$T \cap F$, where an edge and a cut are matched only if the edge crosses the cut.
Every cut $C_i$ is constructed using the edge $e_i'$ and a weight of $w(e_i')$.
Let $e_{\pi(i)}$ be the edge matched to $C_i$, so $e_{\pi(i)}$ crosses $C_i$.

Let $\{e_1, e_2, \ldots, e_{t-1}\} = T \cap F$, and let $e_1', e_2', \ldots, e_{t-1}'$
be the edges of $T'$ that replace the edges in $T \cap F$. The profit of $F$ is exactly
$$
\Delta = \sum_{i=1}^{t-1} w(e_i') - \sum_{i=1}^{t-1} w(e_i) =
\sum_{i=1}^{t-1} \left(w(e_i') - w(e_{\pi(i)})\right),
$$
concluding the proof.
\end{proof}

\begin{proofof}{Lemma~\ref{PROFIT_CLAIM}}
This is a corollary of Claim~\ref{cl: profit claim}. By Lemma~\ref{PROFIT_OF_CUT},
$p_G(C_i) \geq w(e_i') - w(e_{\pi(i)})$ for all $i\in \{1,2,\dots,t-1\}$.
\end{proofof}

\section{Budget Approximation}\label{sec:budget}

In this section we describe the greedy algorithm which chooses cuts with a good ratio of profit to cost.
Then we show that if there exists a solution of cost $B$ and profit $\Delta$, then the greedy algorithm
outputs a solution with profit of at least $\Delta$ and cost of $O(B \cdot \log n)$.

\begin{algorithm}\label{BUDGET_APPROX}
 \SetAlgoLined
 \LinesNumbered
 \DontPrintSemicolon
 \SetKwInOut{Input}{Input}
 \Input{$G=(V, E),\ \Delta$}

 \SetKwRepeat{Do}{do}{while}
 $\budget \gets \min_{e \in E}c(e)$\;
 $\weights \gets \bigcup_{e \in E}^{} \{w(e)\}$\;
\Do{$F = \emptyset$}{
    $F \gets \greedy(G, \budget, \Delta, \weights)$\;
    $\budget \gets 2\cdot \budget$\;
}
\Return{$F$}

\caption{Budget Approximation Algorithm}
\end{algorithm}

\begin{algorithm}\label{GREEDY_ALG}
 \SetAlgoLined
 \LinesNumbered
 \DontPrintSemicolon
 \SetKwInOut{Input}{Input}
 \Input{$G=(V, E),\ \budget,\ \Delta,\ \weights$}

 \SetKwRepeat{Do}{do}{while}
$F \gets \emptyset$, $b \gets 0$, $G'=(V,E') \gets G$\;
\Do{$r^* > 0\wedge b < (1+2\log n)\cdot \budget\wedge p_G(F) < \Delta$}{
    $R \gets []$, $r^* \gets 0$, $C^* \gets \emptyset$\;
    \For{$\{u,v\} \in E'$}{
        \For{$W \in \weights$}{
            $G'' \gets (V,\{e\in E':\ w(e) < W\})$\;
            $C \gets$ minimum $u$-$v$ cut in $G''$, $p \gets W - w(\{u,v\})$\;
            \If{$0 < c(C) \leq \budget \wedge \frac{p}{c(C)} > r^*$}{
                $r^* \gets \frac{p}{c(C)}$, $C^* \gets C$ \;
            }
        }
    }
    $G' \gets G' \setminus C^*$, $b \gets b + c(C^*)$, $F \gets F \cup C^*$\;
}
\eIf{$p_G(F)\ge\Delta$}{
    \Return{$F$}
}{
    \Return{$\emptyset$}
}

\caption{The Greedy Algorithm}
\end{algorithm}

We assume for simplicity that $\Delta > 0$ (otherwise doing nothing is a trivial solution)
and that the cost of a global minimum cut in $G$ (with respect to $c$) is more than $B$
(otherwise, removing a global minimum cut guarantees profit $= \infty\ge\Delta$).
The algorithm proceeds as follows. Start with the lowest possible budget $B$, the minimum
cost of a single edge, and search for $B$ using spiral search, doubling the guess in each
iteration. The test for $B$ is to get profit at least $\Delta$, where the cost of the solution
is restricted to be less than $(1+2\log n)\cdot B$. Clearly, if the test gives the correct answer, we will
overshoot $B$ by a factor of less than $2$, and therefore we will pay for a solution with a profit
of at least $\Delta$ a cost of less than $(2+4\log n)\cdot B = O(B\log n)$. See the pseudo-code
of Algorithm~\ref{BUDGET_APPROX}.

Now, for a guess of $B$, we repeatedly find a partial cut of cost $\le B$ with maximum
profit-to-cost ratio, and remove it, until we either fail to make progress, or exceed the
relaxed budget, or accumulate a profit of at least $\Delta$. In the latter case, we've
reached our target and can stop searching for $B$. To find the best partial cut, we
enumerate over the edges of the graph and over the possible weights of edges of
the graph. For an edge $e$ and a weight $W$, we consider the minimum cost cut
that separates the endpoints of $e$, taking into account only edges of weight less
than $W$. That is, we consider in the current graph $G'$ (after previously chosen
cuts have been removed) the cheapest cut $C_{G'}(S,W)$ with $|S\cap e| = 1$, and
we choose among those cuts for all $e,W$ a cut with the maximum profit-to-cost
ratio. See the pseudo-code of Algorithm~\ref{GREEDY_ALG}.

\begin{theorem}\label{BUDGET_APPROX_THEOREM}
Suppose that there exists a solution of cost $B$ and profit $\Delta$. If the input
budget in Algorithm~\ref{GREEDY_ALG} is at least $B$, then the output $F$
of the algorithm has $p_G(F)\ge\Delta$.
\end{theorem}

We begin with the analysis of a single iteration of Algorithm~\ref{GREEDY_ALG}.
\begin{lemma}\label{SINGLE_ITERATION_RATIO}
Consider an iteration of the do-loop in Algorithm~\ref{GREEDY_ALG}
(with input budget $B$ and input target profit $\Delta$). Suppose that $G'$
has a solution $F$ of cost $c(F) \le B$ and profit $p_{G'}(F)=\Delta - \delta$.
Then, this iteration computes a partial cut $C = C_{G'}(S,W)$ that satisfies,
for some $e\in C$,
\[\frac{W-w(e)}{c(C)} \geq \frac{\Delta - \delta}{2B\cdot \log n}.\]
\end{lemma}

\begin{proof}
By Claim~\ref{cl: profit claim} and Lemma~\ref{SUMS_CLAIM}, there are partial
cuts $C_1, C_2, \ldots, C_{t-1}$ in $G'$, edges $e_1\in C_1$, $e_2\in C_2$, $\dots$,
$e_{t-1}\in C_{t-1}$, and edge weights $W_1,W_2,\ldots, W_{t-1}$ defining the cuts,
such that
\[\sum_{i=1}^{t-1} W_i - w(e_i) \geq \Delta,\]
and
\[\sum_{i=1}^{t-1} c(C_i) \leq 2B\cdot \log n.\]
Hence, there exists a cut $C_i$ with a ratio
\[\frac{W_i-w(e_i)}{c(C_i)} \geq \frac{\Delta-\delta}{2B\cdot \log n}.\]
Moreover, by Claim~\ref{cl: cut in G}, $C_i\subset F$ and hence $c(C_i) \leq B$.

The algorithm iterates over all the edges and over all the weights. Consider the
iteration that uses the weight $W_i$ and the edge $e_i$. Let $C'$ be the cut that
the algorithm finds in this iteration.

Notice that $C'$ is a minimum cost cut separating the endpoints of $e_i$ in the
subgraph of $G'$ consisting of edges of weight $< W_i$. Therefore,
$c(C')\le c(C_i)$. As $e_i$ crosses $C'$ and the weight defining $C'$ is
$W_i\ge w(e_i)$, we get that
\[\frac{W_i-w(e_i)}{c(C')} \ge \frac{W_i-w(e_i)}{c(C_i)} \geq \frac{\Delta-\delta}{2B\cdot \log n}.\]
The algorithm's choice of $e,W,C$ maximizes the expression $\frac{W-w(e)}{c(C)}$, hence
the lemma follows.
\end{proof}

\begin{proofof}{Theorem~\ref{BUDGET_APPROX_THEOREM}}
Let $F^*$ be the promised solution with cost $c(F^*) = B$ and profit $p_G(F^*) = \Delta$.
Let $B'\ge B$ denote the budget that Algorithm~\ref{GREEDY_ALG} gets as input.
Let $C_1 = C_G(X_1,W_1)$, $C_2 = C_{G\setminus C_1}(X_2,W_2)$, $\ldots$,
$C_l = C_{G\setminus (C_1\cup C_2\cup \cdots \cup C_{l-1})}(X_l,W_l)$ be the
sequence of partial cuts that the algorithm chooses. Clearly,
$c(\bigcup_{i=1}^{l}C_i) < (2+2\log n)\cdot B'$, on account of the stopping
conditions of the do-loop (the last iteration started with total cost below
$(1+2\log n)\cdot B'$). We need to show that $p_G(\bigcup_{i=1}^{l}C_i) \ge \Delta$.
This is clearly the case if the do-loop terminates because $p_G(F) \ge \Delta$, so
we need to exclude the other termination conditions. Note that by
Lemma~\ref{SINGLE_ITERATION_RATIO}, if $p_G(F) < \Delta$, then
$r^* > 0$, hence we only need to show that $b$ does not reach or exceed
$(1+2\log n)\cdot\budget$ before $p_G(F)$ reaches or exceeds $\Delta$.

For every $i\in\{1,2,\dots,l-1\}$, denote $C_i^F = C_i\cap F^*$ and $C_i^H = C_i\setminus F^*$.
Also put $S_i = \bigcup_{j=1}^{i} C_j$, $S_i^F = \bigcup_{j=1}^{i} C_j^F$, and
$S_i^H = \bigcup_{j=1}^{i} C_j^H$.
Let $h_i$ be a minimum weight edge in $C_i^H$ and $e_i'$ be a minimum weight edge
in $C_i$. Put $p_i = W_i - w(h_i)$. Consider $G_i = G \setminus S_{i-1}$, namely the graph
just before the algorithm chooses $C_i$. Notice that the criterion for choosing $C_i$
involves the ratio
\[\frac{W_i-w(e_i')}{c(C_i)} = \frac{w(h_i)-w(e_i') + p_i}{c(C_i^F) + c(C_i^H)}.\]

Consider the partial cut $C_i' = \{e \in C_i | w(e) < w(h_i)\}$. As no edge of $C_i^H$ has weight
less than $w(h_i)$, it must be that $C_i'\subseteq C_i^F$. Therefore, $c(C_i')\le c(C_i^F)$, and
hence $\frac{w(h_i) - w(e_i')}{c(C_i')} \ge \frac{w(h_i)-w(e_i')}{c(C_i^F)}$.
However, it must be that $\frac{w(h_i) - w(e_i')}{c(C_i')} \le \frac{W_i-w(e_i')}{c(C_i)}$, otherwise
the algorithm would choose $C_i'$ instead of $C_i$. Therefore, we have that
\[\frac{W_i-w(e_i')}{c(C_i)} = \frac{w(h_i)-w(e_i') + p_i}{c(C_i^F) + c(C_i^H)} \ge \frac{w(h_i)-w(e_i')}{c(C_i^F)}.\]
This implies that
\[\frac{p_i}{c(C_i^H)} > \frac{w(h_i)-w(e_i') + p_i}{c(C_i^F) + c(C_i^H)} = \frac{W_i-w(e_i')}{c(C_i)}\]
(as $\frac{a+b}{c+d} \ge \frac{a}{c}\implies \frac{b}{d} \ge \frac{a + b}{c + d}$ for $a,b\ge 0$ and $c,d > 0$).

Denote $F' = F^* \cap S_l$. We now show that $p_G(S_l) \ge p_G(F') + \sum_{i=1}^{l}p_i$. To
estimate $p_G(S_l)$, notice that $G \setminus S_l = G \setminus F' \setminus (S_l^H)$. Therefore,
we can lower bound the profit by summing over $i=1,2,\dots,l$ the profit of removing $C_i^H$
from $G\setminus F'\setminus S_{i-1}^H$. Let
$G_i' = G \setminus F' \setminus (S_{i-1}^h) = G \setminus F' \setminus (S_{i-1})$. This graph is
$G_i$ minus some edges from $F^*$. As $h_i \in C_i^H$, it must be that $h_i \notin F^*$ so $h_i$
is in $G_i'$. Now, $C_i^H$ includes the edge $h_i$ of weight $w(h_i)$. Furthermore, none of the
edges of $C_i$ are in $G_{i+1}'$, so $C_i^H = C_{G_i'}(X_i,W')$, for $W'\ge W_i$.
Therefore, by Lemma~\ref{PROFIT_OF_CUT}, $p_{G_i'}(C_i^H) \ge W_i - w(h_i)$. We conclude
that
\begin{equation}\label{eq: pG(Sl)}
p_G(S_l) = p_G(F')+ \sum_{i=1}^{l} p_{G_i'}(C_i^H)\ge p_G(F') + \sum_{i=1}^{l}p_i.
\end{equation}

Put $\delta = p_G(F') = p_G(F^* \cap S_l)$. By Corollary~\ref{PROFIT_PRESERVATION},
as $G_i = G \setminus S_{i-1}$ and $F^* \cap S_{i-1}\subseteq S_{i-1}$, we have
\begin{equation}\label{eq: p_{G_i}(F* setminus Si-1)}
p_{G_i}(F^* \setminus S_{i-1}) \ge p_{G \setminus (F^* \cap S_{i-1})}(F^* \setminus S_{i-1}).
\end{equation}
Write
\begin{equation}\label{eq: pG(F*)}
p_G(F^*) = p_G(F^* \cap S_{i-1}) + p_{G \setminus (F^* \cap S_{i-1})}(F^* \setminus S_{i-1}).
\end{equation}
Combining Equations~\eqref{eq: p_{G_i}(F* setminus Si-1)} and~\eqref{eq: pG(F*)}, we get
$$
p_{G_i}(F^* \setminus S_{i-1}) \ge p_{G \setminus (F^* \cap S_{i-1})}(F^* \setminus S_{i-1}) =
p_G(F^*) - p_G(F^* \cap S_{i-1}) \ge p_G(F^*) - p_G(F^* \cap S_l) = \Delta - \delta,
$$
where the last inequality follows from the fact that $F^* \cap S_{i-1}\subseteq F^* \cap S_l$.

By Lemma~\ref{SINGLE_ITERATION_RATIO}, as there exists in $G_i$ a solution
$F^* \setminus S_{i-1}$ with profit $p_{G_i}(F^* \setminus S_{i-1})\ge \Delta - \delta$
and cost $c(F^* \setminus S_{i-1})\le B\le B'$, it holds that there exists a partial cut $C=C_{G_i}(S,W)$
and an edge $e\in C$, such that $\frac{W-w(e)}{c(C)}\ge \frac{\Delta - \delta}{2B'\log n}$.
In particular, $C_i$ must satisfy this inequality, as it maximizes the left-hand side. Therefore,
\[\frac{p_i}{c(C_i^H)} \ge \frac{W_i-w(e_i')}{c(C_i)} \ge \frac{\Delta - \delta}{2B' \log n}.\]
So, $p_i \ge c(C_i^H)\cdot \frac{\Delta - \delta}{2B' \log n}$. Plugging this into Equation~\eqref{eq: pG(Sl)},
we get that
\begin{equation}\label{eq: final profit bound}
p_G(S_l) \ge p_G(F') + \sum_{i=1}^{l}p_i \ge \delta + \sum_{i=1}^{l}c(C_i^H)\cdot \frac{\Delta - \delta}{2B' \log n} =
\delta + c(S_l^H)\cdot \frac{\Delta - \delta}{2B' \log n}.
\end{equation}
Now, we assumed that the do-loop does not terminate because $p_G(F) \ge \Delta$, so it must
have terminated because $b \ge (1+2\log n)\cdot B'$. Therefore,
$c(S_l^H) = c(S_l) - c(F') \ge c(S_l) - c(F^*)\ge B' + 2B' \log n - B\ge 2B'\log n$, hence
the right-hand side of Equation~\eqref{eq: final profit bound} is at least $\Delta$.
\end{proofof}

\paragraph{Running time.}
Recall that $\tau(n,m)$ denotes the time complexity of computing a minimum $s$-$t$ cut, where
$n$ is the number of nodes of the network and $m$ is the number of edges of the
network. Let $d=|\weights|$ denote the number of different edge weights (notice that $d\le m$).
The doubling search for the right budget adds a factor of $O\left(\log\frac{B}{c_{min}}\right)$.
Each iteration of the do-loop in Algorithm~\ref{GREEDY_ALG} iterates over all the edges
and the weights, and executes one minimum $s$-$t$ cut computation, so the time
complexity of a do-loop iteration is $O(\tau(n,m) \cdot dm)$. In each such iteration
we remove at least one edge, so there are no more than $m$ iterations of the do-loop.
Therefore, the total running time of the algorithm is
\[O\left(\tau(n,m) \cdot dm^2\log\frac{B}{c_{min}}\right).\]

It is possible to reduce the factor of $\log\frac{B}{c_{min}}$ to $\log m$ by
reducing the range of the search for $B$ as follows. With a budget of $B$,
we cannot remove edges of cost $> B$. Therefore,
$$
b^* = \arg\min\{b:\ \MST(G\setminus\{e\in E:\ c(e)\le b\})\ge\MST(G)+\Delta\}
$$
is a lower bound on $B$. On the other hand, by removing all the edges of cost
at most $b^*$ we definitely gain $\Delta$. There are at most $m$ such edges,
so $mb^*$ is an upper bound on $B$.

It is possible to improve the running time to
\[O\left(\tau(n,m) \cdot dm\log m\right)\]
using a more clever implementation as follows. Firstly, we calculate the cuts for each weight
and edge only in the first iteration. In the following iterations we can use the same set of cuts,
ignoring the cuts that were created using the edges that we already removed. We need to
show that a version of claim~\ref{SINGLE_ITERATION_RATIO} holds for this faster algorithm.
\begin{claim}
Consider an iteration of the do-loop of Algorithm~\ref{GREEDY_ALG}, assuming that
the input $\budget\ge B$. Let $\delta = \MST(G') - \MST(G)$. Consider the cuts
computed during the first do-loop iteration (i.e., partial cuts in $G$), in an iteration of
the nested for-loop with $e\in E\setminus F$ and $W\in\weights$. Let $C$ be such a
cut with the best ratio $\frac{W-w(e)}{c(C)}$ (in $G$). Then,
\[\frac{W-w(e)}{c(C\cap E')} \geq \frac{\Delta - \delta}{2B\cdot \log n}.\]
\end{claim}

\begin{proof}
Let $F^*$ denote an optimal solution with a budget $B$, yielding an increase $\Delta$
in the weight of a minimum spanning tree. Consider the ``intermediate" graph
$G'' = G\setminus (F\cap F^*)$. Notice that $\MST(G'') - \MST(G)\le\delta$, so
$F'' = F^*\setminus F$ is a solution in $G''$ that costs less than $B$ and gains at least
$\Delta-\delta$. The same bounds on cost and gain holds also in $G'$.
By Theorem~\ref{CUTS_THEOREM}, there exist partial cuts $C''_1, \ldots C''_{t-1}$, in
$G''$, $C''_i = C_{G''}(S_i,W_i)$ for all $i$, such that the following inequalities hold.
$\sum_i c(C''_i)\le 2 \cdot B \cdot \log n$, and $p_{G''}(\cup_i C''_i)\ge \Delta-\delta$.
Moreover, by Claim~\ref{cl: cut in G}, $\cup_i C''_i\subset F''$, and by
Claim~\ref{cl: profit claim}, there are edges $e_i\in C''_i$, for all $i$ such that
$\sum_{i=1}^{t-1} (W_i - w(e_i))\ge\Delta-\delta$. Consider the cuts $C_i = C_G(S_i,W_i)$
in $G$, for all $i$. The latter inequality clearly holds. It is also true that
$\sum_i c(C_i)\le 2B \cdot \log n$, because
$\cup_i (C_i\setminus C''_i)\subset F^*\setminus F''$, $c(F^*) = B$, and by
Claim~\ref{cl: num cuts}, every edge is contained in at most $2\log n$ cuts.
Therefore, there exists $i$ for which
$\frac{W_i - w(e_i)}{c(C_i)}\ge \frac{\Delta - \delta}{2B\cdot \log n}$.
The cut $C$ computed in the first iteration of the do-loop (for $G$) for the
choice $e_i$ and $W_i$ has $c(C)\le c(C_i)$, hence
$\frac{W_i - w(e_i)}{c(C)}\ge \frac{\Delta - \delta}{2B\cdot \log n}$. As
$e_i\not\in F$, we consider $C$ in the iteration for $G'$. As
$c(C\cap E')\le c(C)$, we have
$\frac{W_i - w(e_i)}{c(C\cap E')}\ge \frac{\Delta - \delta}{2B\cdot \log n}$,
as claimed.
\end{proof}

The rest of the proof is the same, so the faster algorithm keeps the approximation guarantees.
Algorithm~\ref{GREEDY_ALG} uses $dm$ flow calculations in the first iteration of the do-loop.
Subsequence iterations do not require additional flow calculations, only enumeration over at
most $dm$ cuts computed in the first iteration. Therefore, as $\tau(n,m) = \Omega(m)$, we get
that the running time is
\[O((\tau(n,m) \cdot dm + dm^2)\log m) = O(\tau(n,m) \cdot dm\log m).\]

\section{Profit Approximation}\label{sec:profit}

In this section we discuss the profit problem.

\subsection{Profit approximation algorithm}\label{sec:profit approx}

It is possible to use our methods to achieve $O(\log n)$-approximation to the profit given a strict budget
$B$, a problem considered previously in~\cite{FS96,BTV11,Zen15,LS17}. In comparison to previous
work, our results approximate the {\em profit} (i.e., the increase in minimum spanning tree weight) of
interdiction, rather than the total weight of the final minimum spanning tree. Clearly, these results are
incomparable to the claims of previous work. We demonstrate in the following subsection that the algorithm
in~\cite{LS17} does not provide any approximation guarantee for the increase in weight rather than the total 
weight.

The algorithm in this case 
is based on Algorithm~\ref{BUDGET_APPROX}, with the following changes.
Firstly, in each iteration take a cut with the best ratio among the cuts that do not cause the cost to exceed $B$.
Stop when there are no cuts we can take without exceeding the budget. Secondly, take the best solution
between this option and taking just one cut in $G$ that has maximum profit subject to the budget constraint.
Notice that this algorithm resembles the greedy approach to approximating the knapsack problem.

\begin{algorithm}\label{PROFIT_APPROX}
 \SetAlgoLined
 \LinesNumbered
 \SetKwInOut{Input}{Input}
 \Input{$G=(V, E)$, $B$, $\Delta$, $\weights$}

 \SetKwRepeat{Do}{do}{while}

$R \gets []$, $p_0 \gets 0$, $C_0 \gets \emptyset$\;
    \For{$\{u,v\} \in E$}{
        \For{$W \in \weights$}{
            $G'' \gets (V,\{e\in E:\ w(e) < W\})$\;
            $C \gets$ minimum $u$-$v$ cut in $G''$\;
            \If{$c(C) \leq B$ and $p_{G}(C) > p_0$}{
                $p_0 \gets p_{G}(C)$, $C_0 \gets C$\;
            }
        }
    }
$G'=(V,E') \gets G$, $F \gets \emptyset$, $b \gets 0$\;
\Do{$r^* > 0$}{
    $R \gets []$, $r^* \gets 0$, $C^* \gets \emptyset$\;
    \For{$\{u,v\} \in E'$}{
        \For{$W \in \weights$}{
            $G'' \gets (V,\{e\in E':\ w(e) < W\})$\;
            $C \gets$ minimum $u$-$v$ cut in $G''$, $p \gets W - w(\{u,v\})$\;
            \If{$b < b + c(C) \leq B\wedge \frac{p}{c(C)} > r^*$}{
                $r^* \gets \frac{p}{c(C)}$, $C^* \gets C$ \;
            }
        }
    }
    $G' \gets G' \setminus C^*$, $b \gets b + c(C^*)$, $F \gets F \cup C^*$\;
}
\eIf{$p_{G}(C_0) \geq p_{G}(F)$}{
    \Return{$C_0$}
}{
    \Return{$F$}
}

\caption{Profit Approximation Algorithm}
\end{algorithm}

\begin{theorem}
If there exist a solution $F$ with profit $\Delta$ and cost $B$, then Algorithm~\ref{PROFIT_APPROX}
computes a solution that costs at most $B$ and gives a profit of at least
\[\frac{\Delta}{4} \cdot \left(\frac{1}{\log n} - \frac{1}{\log^2 n}\right) = \Omega\left(\frac{\Delta}{\log n}\right).\]
\end{theorem}

\begin{proof}
We will refer to the loop that computes $p_0$ and $C_0$ as the first phase of the algorithm, and
to the other loop at the second phase of the algorithm. Assume that in the second phase the algorithm
already chose to remove the edges $H \subseteq V$ and increased the minimum spanning tree weight
by $p_G(H) = \delta$. Let $G' = (V,E') = G \setminus H$. We show that if $c(H) \leq \frac{B}{2}$, then at
least one of the following cases is true.
\begin{enumerate}
    \item\label{it: cheap cut} In the next iteration the algorithm enumerates over a cut $C'$ with a ratio
    of at least $r_{G'}(C') \geq \frac{\Delta-\delta}{2B\cdot \log n}$ and cost $c(C') \leq \frac{B}{2}$.

    \item\label{it: profitable cut} In the first phase, the algorithm enumerates over a cut $C$ in the original
    graph $G$ with a profit of $p_{G}(C) \geq \frac{\Delta-\delta}{4\cdot \log n}$ and cost $c(C) \leq B$.
\end{enumerate}

Denote
$\bar{G} = G \setminus (H \cap F)$ and $c(H\cap F) = b$. Consider the set of edges $F \setminus H$.
Clearly,
$$
c(F \setminus H) = c(F) - c(H\cap F) = B-b.
$$
Moreover, $p_G(F) = \Delta$, whereas $p_G(H \cap F) \leq p_G(H) = \delta$, so $p_{\bar{G}}(F \setminus H)\ge \Delta - \delta$.
Using Claim~\ref{cl: profit claim} and the first assertion of Theorem~\ref{CUTS_THEOREM}, there exists
a partial cut $\bar{C} = C_{\bar{G}}(S,W)\subset F \setminus H$ in $\bar{G}$ and an edge $e\in \bar{C}$,
such that $c(\bar{C})\le B-b$ and
$$
\frac{W - w(e)}{c(\bar{C})} \ge \frac{\Delta-\delta}{2(B-b)\cdot \log n}.
$$

Now, consider the same partition in $G'$, i.e., the partial cut $C' = C_{G'}(S,W)$. As $G' = \bar{G}\setminus (H\setminus F)$,
we have that $c(C')\le c(\bar{C})\le B-b$. Moreover, as $e\in F\setminus H$, also $e\in E'$. Hence, by Lemma~\ref{PROFIT_OF_CUT},
$p_{G'}(C')\ge W - w(e)$. Therefore,
$$
r_{G'}(C') = \frac{p_{G'}(C')}{c(C')}\geq \frac{W - w(e)}{c(C')}\ge \frac{W - w(e)}{c(\bar{C})}\ge \frac{\Delta-\delta}{2(B-b)\cdot \log n},
$$
where the first inequality uses Lemma~\ref{PROFIT_OF_CUT}.

If $c(C') \leq \frac{B}{2}$ then Case~\ref{it: cheap cut} holds.
Otherwise, $c(C') > \frac{B}{2}$ and
$$
W - w(e) =  c(C')\cdot \frac{W - w(e)}{c(C')} > \frac{B}{2}\cdot \frac{\Delta-\delta}{2(B-b)\cdot \log n} \geq \frac{\Delta-\delta}{4 \log n}.
$$
Therefore, using Lemma~\ref{PROFIT_OF_CUT} again,
$$
p_{G}(C) \geq W - w(e)\ge \frac{\Delta-\delta}{4 \log n}.
$$
Moreover, $c(C)\le c(\bar{C}) + c(H \cap F)\le B$, so Case~\ref{it: profitable cut} holds.

Consider the second phase of the algorithm, and the first iteration that begins with
Case~\ref{it: cheap cut} not holding. If the current profit $\delta\ge \frac{\Delta}{\log n}$,
we are done. Otherwise, if the current total cost $b > \frac{B}{2}$, then the following
holds. All previous iterations started with $b\le \frac{B}{2}$, hence each added to
the solution a partial cut with profit to cost ratio of at least $\frac{\Delta-\Delta/\log n}{B\cdot \log n}$.
So the total profit is at least
$$
\frac{B}{2}\cdot \frac{\Delta-\Delta/\log n}{B\cdot \log n} = \frac{\Delta}{2}\cdot\left(\frac{1}{\log n} - \frac{1}{\log^2 n}\right).
$$
The remaining case is that $\delta < \frac{\Delta}{\log n}$, $b\le\frac{B}{2}$, and Case~\ref{it: cheap cut}
does not hold. But then Case~\ref{it: profitable cut} must hold. Hence,
$$
p_0\ge \frac{\Delta}{4}\cdot \left(\frac{1}{\log n} - \frac{1}{\log^2 n}\right),
$$
thus completing the proof.
\end{proof}

Notice that in the uniform removal costs case it holds that $t \le B$, and therefore the same proof shows
a profit of $\Omega\left(\frac{\Delta}{\log B}\right)$.

\paragraph{Running time.}
The analysis is very similar to that of the budget approximation algorithm.
Each of at most $m$ do-loop iterations iterates over edges and weights
$O(dm)$ times, where $d$ denotes the number of different edge
weights. Each internal iteration computes a minimum $s$-$t$ cut, in time
$\tau(n,m)$. Thus, the total running time is $O(\tau(n,m) \cdot dm^2)$.
With the same modification of calculating the cuts only in the first iteration,
it is possible to achieve the same asymptotic approximation guarantees
while improving the time complexity to $O(\tau(n,m) \cdot dm)$.

\subsection{Bad example for previous algorithms}\label{sec:profit_example}

When the optimal increase is small relative to the weight of the initial minimum spanning tree, our approximation guarantees are stronger than the constant factor approximations of the final tree weight. In order to demonstrate that this actually happens with previous algorithms, we analyze an instance that is motivated by the NP-hardness reduction for spanning tree interdiction in~\cite{FS96}. The constant approximation convex optimization-based algorithms, such as~\cite{Zen15,LS17}, fail to give any non-trivial solution for this example.

Let $G_H = (V_H, E_H)$ be an instance of the maximum components problem defined in~\cite{FS96} for that maximum number of connected components that can be created by removing $B$ edges from $G_H$ is $b$. Construct a graph $G = (V, E)$ by adding to $G_H$ four new vertices, as follows. Set $V=V_H \cup \{t_1, t_2, t_3, t_4\}$, $E=E_H \cup \{(u,v_1) | u \in V_H\} \cup E_T$, where $E_T$ are the edges between the new vertices as explained later. Assign weights $w=0$ and removal costs $r=1$ to all edges in $E_H$, and $w=1, r=\infty$ (where $\infty$ is some constant above $B+1$) to the edges between $G_H$ and $v_1$. The edges in $E_T$ are as follows: $(v_1, v_2)$ with $w=0, r=\infty$, $(v_1, v_3)$ with $w=W, r=\infty$, $(v_2, v_3)$ with $w=0, r=B+1$, $(v_1, v_4)$ with $w=W+\frac 1 2, r=\infty$, and $(v_2, v_4)$ with $w=W, r=\frac 1 2$.

The initial minimum spanning tree of $G$ has weight of $W+1$. We consider as instance of the profit maximization problem on $G$ with a budget of $B+\frac 1 2$. An optimal solution for $G$ with this budget has spanning tree weight of $W+b+\frac 1 2$, thus the profit is $\Delta = b-\frac 1 2$. It is obtained by removing $(v_2,v_4)$ in addition to the $B$ edges of the optimal maximum components solution in $G_H$. Notice that by spending a cost of $B+1$ (which exceeds the budget), it is possible to remove $(v_2,v_3)$, and obtain a spanning tree with weight of $2W + 1$.

To demonstrate our claim, we analyze the performance of the algorithm in~\cite{LS17} on this example. The conclusion holds also for other similar methods, such as the one in~\cite{Zen15}. We choose a sufficiently large value $W > B+1$.
With budget $B+\frac 1 2$, the algorithm finds two integer solutions $R_1, R_2$ as follows: $R_1$ is the ``empty" solution $(w=0, MST=W+1)$, and $R_2$ is the over-budget solution $(w=B+1, MST=2W+1)$ obtained by removing $(v_2,v_3)$. Notice that the ``bang-per-buck" of $R_2$ is $\frac{W}{B+1} > 1$. For any other solution $R'$ so that $(v_2,v_3) \notin R'$, it is guaranteed that bang-per-buck is not greater than $1$ as the profit from removing any other edge cannot exceed its cost (either for $(v_2,v_4$) or any edges set in $G_H$). As there is no other solution above the connecting line between $R_1,R_2$ (and no other more expensive relevant solution), these solutions are two optimal solutions of the Lagrangian relaxation, for the Lagrange multiplier $\lambda = \frac{W}{B+1}$.

\usetikzlibrary{shapes.geometric, intersections}
\begin{figure}
\centering
\tikzset{mynode/.style={draw, very thick, circle, minimum size=1cm}, myarrow/.style={very thick}}

\begin{tikzpicture}
\node[mynode](v1) at (42:5.7){$v_1$};
\node[mynode](v2) at (3:4.3){$v_2$};
\node[mynode](v3) at (20:11){$v_3$};
\node[mynode](v4) at (0.5:10.3){$v_4$};

\node[ellipse, draw, very thick, minimum width = 3cm, minimum height = 3cm, dashed] (G_H) at (-2,2.5){$G_H$};

\draw[myarrow](v1)--(v4) node[midway, sloped, above, xshift=47pt] {\footnotesize $w=W+0.5, r=\infty$};
\draw[myarrow](v1)--(v2) node[midway, sloped, below] {\footnotesize $w=0, r=\infty$};
\draw[myarrow](v1)--(v3) node[midway,above] {\footnotesize $w=W, r=\infty$};
\draw[myarrow](v2)--(v4) node[midway,below] {\footnotesize $w=W, r=0.5$};
\draw[myarrow](v2)--(v3) node[midway, sloped, above, xshift=-48pt] {\footnotesize $w=0, r=B+1$};

\node (text) at (-2.1,0.5) {\footnotesize $w=0, r=1$};
\node (text) at (1.1,4) {\footnotesize $w=1, r=\infty$};

\draw[myarrow] (-1.5,3.7) -- (v1);    
\draw[myarrow] (-1.2,3.3) -- (v1);  
\draw[myarrow] (-1.3,1.3) -- (v1);  
\node at (1.2,3) {\Large $\vdots$}; 

\end{tikzpicture}
\caption{Bad Example for Previous Algorithms}
\end{figure}
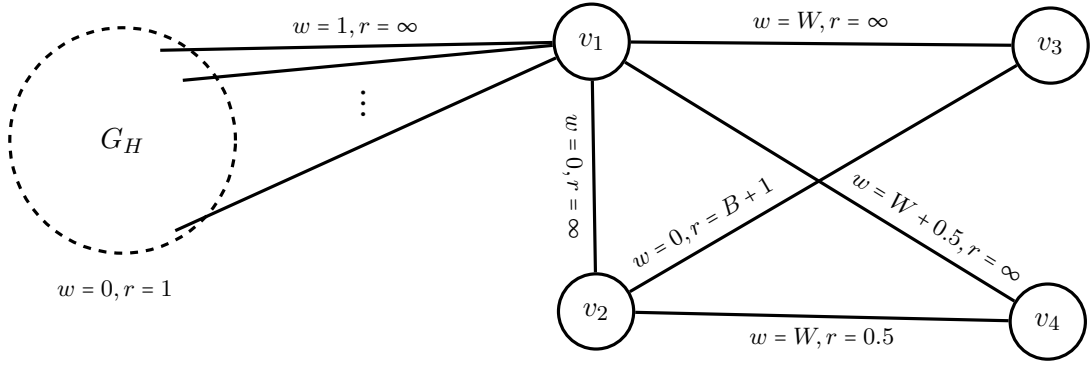

The algorithm chooses the best among the three options:
\begin{enumerate}
    \item Return a spanning tree with weight of at least $w_k=W + \frac 1 2$ (the smallest weight so that the graph without heavier edges is still connected under any removal of edges within the budget $B+\frac 1 2$). In our example this solution is obtained by removing $(v_2,v_4)$ so the MST weight is $W+\frac 3 2$.
      
    \item Return the empty solution $R_1$ (yielding the original minimum spanning tree of weight $W+1$).

    \item Return $R$, the trimmed version of $R_2$. The solution $R$ is created using a reduction to tree knapsack. It holds that $R \subset R_2$, and the cost of $R$ is no more than $B+\frac 1 2$. As $R_2 = \{(v_2,v_3)\}$, the only subset that does not exceed the profit is $R=\emptyset$, which again produces the trivial empty solution.
\end{enumerate}

Therefore, in our example the algorithm of~\cite{LS17} chooses the first option, obtaining a solution with spanning tree weight of $W+\frac 3 2$. As the optimal solution is $W+b+\frac 1 2$ (and $W > B+1 \geq b$), the algorithm indeed achieves the promised constant factor guarantee against the total cost of the tree. However, the algorithm achieves a profit of only $\frac 1 2$. The optimal profit is $b-\frac 1 2$, which can be arbitrarily large compared to $\frac 1 2$, depending on maximum components solution in $G_H$.

\section{The \boldmath$\varepsilon$-Protection Problem}\label{sec:def}

The analysis in Section~\ref{sec:eps-inc} implies a good defense against $\varepsilon$-increase.
Before presenting
the algorithm, we first formalize the problem. The input is a graph $G = (V,E)$, another set of edges
$E'$ over $V$, edge weights $w:E\cup E'\rightarrow \mathbb{R}^+$, edge construction
costs $b:E'\rightarrow \mathbb{R}^+$, and edge removal costs $c:E\cup E'\rightarrow \mathbb{R}^+$.
For a graph $G$, let $F^*(G)$ denote an optimal solution to the $\varepsilon$-increase problem
discussed above. Our goal in the $\varepsilon$-protection problem is to compute a set of edges
$S\subset E'$ to add to $G$ so that $c(F^*(G\cup S)) > c(F^*(G))$, minimizing the building cost $b(S)$.

We assume that adding any edge $e\in E'$ to $G$ does not reduce the weight of a minimum spanning
tree. Also, we allow parallel edges (so, for instance, pairs of nodes may be connected by an edge in $E$
and also by an edge in $E'$).

Based on the algorithm for $\varepsilon$-increase here is a simple approximation algorithm for
this problem. The first step is to list all the partial cuts that the $\varepsilon$-increase
algorithm considers, which have optimal cost. Notice that every cut that the algorithm computes
is derived from a global minimum cut of a subgraph of $G$. In that subgraph, there are at most
$\binom{|V|}{2}$ global minimum cuts, and those cuts can be enumerated efficiently. The number
of subgraphs to consider is $n-1$. Thus, the number of listed cuts is less than $n^3$. We want to
add at least one edge from $E'$ to every listed partial cut. It is possible to approximate an optimal
solution within a factor of $O(\log n)$ using the greedy approximation for weighted SET COVER.
Simply, associate with each edge in $E'$ the set of partial cuts it increases their cost, and then
approximate the minimum $b$-weight set of edges that covers all listed cuts.

\appendix
\section{Proofs Appendix}
\label{appendix:A}

\begin{proofof}{Lemma~\ref{MST_DELETED_EDGES_CLAIM}}
In this proof, we will use the so-called {\em blue rule}, which states the following:
suppose you have a graph $G$ and some of the edges of a minimum spanning
tree are colored blue. If you take any complete cut $C$ of $G$ that contains no blue
edge, and any edge $e \in C$ of minimum weight, then there exists a minimum spanning
tree of $G$ that contains all the blue edges and $e$.

Consider the edges $e \in T \setminus F$ in arbitrary order. The complete cuts $C_e = C_{T,e}$
are disjoint. Also, such an edge $e$ has minimum weight in $C_e$, and therefore also minimum
weight in $C_e\setminus F$. Thus, we can use the blue rule repeatedly in $G'$ to color all the
edges in $T \setminus F$ blue.
\end{proofof}

\begin{proofof}{Lemma~\ref{PROFIT_OF_CUT}}
If $W \leq w(e)$ then the claim is trivial, as the profit of a cut is non-negative. Thus, we may
assume that $W > w(e)$. Moreover, the worst case is when $e$ has minimum weight in $C_G(S)$,
because if the claim holds for a minimum weight edge then it holds also for all edges.

Clearly, if $e=(u,v)$ is a minimum weight edge in $C_G(S)$, then there exists a minimum spanning
tree $T$ of $G$ that contains $e$ (apply the blue rule to $C_G(S)$ and $e$).

Let $T$ be a minimum spanning tree of $G$ satisfying $e\in T$, and let $T'$ be a minimum
spanning tree of $G \setminus C$. As $e \in C_G(S)$ and $W > w(e)$, it holds that $e \notin T'$.
By adding $e$ to the tree $T'$, we create a cycle $P$. As $e\in P$ crosses $C(S)$, there
must be another edge $e'\in P$ that crosses $C_G(S)$. Clearly, $e'\in T'$ because the only
edge in $P\setminus T'$ is $e$. It holds that $w(e') \geq W$, because otherwise $e'\in C$
and therefore not in $T'$.

Assume for contradiction that $c(T') < c(T) + W - w(e)$. Replacing $e'$ with $e$ we create
a new spanning tree $T''$ of $G$ of weight
\[c(T'') \leq c(T') - w(e') + w(e) \leq c(T') - W + w(e) < c(T),\]
a contradiction to the fact that $T$ is a minimum spanning tree of $G$.
Thus, it holds that $c(T') \geq c(T) + W - w(e)$, and therefore $p_G(C)\ge W - w(e)$.
\end{proofof}

\begin{proofof}{Lemma~\ref{SUPERMODULARITY}}
If $G \setminus B \setminus \{e\}$ is not connected, then as $G \setminus B$ is connected by
assumption, we have $p_{G \setminus B}(e) = \infty \ge p_G(e)$, so the lemma holds. Thus, we
may assume that $G \setminus B \setminus \{e\}$ (and therefore also $G \setminus \{e\}$) is
connected.

Let $e = \{u,v\}$. We set $W$ to be the maximum over all $u$-$v$ cuts in $G \setminus \{e\}$
of the minimum weight edge crossing the cut. More formally,
$$
W = \max\{ \min \{ w(e'):\ e'\in C_{G \setminus \{e\}}(S) \}:\ S\subset V\wedge |\{u,v\} \cap S| = 1\}.
$$
We show that if $W \ge w(e)$, then $p_{G}(e) = W - w(e)$.
By Lemma~\ref{PROFIT_OF_CUT} we have $p_{G}(e) \ge W - w(e)$, so it suffices to prove the
reverse inequality.

Let $T$ be an arbitrary minimum spanning tree of $G$. If $e\not\in T$, then every edge $e'$ on
the path in $T$ connecting $u$ and $v$ must have $w(e')\le w(e)$, hence $W\le w(e)$ and the
claim holds vacuously if $W < w(e)$ and as $p_G(e) = 0$ if $W = w(e)$.
Otherwise, if $e\in T$, then
by Lemma~\ref{MST_DELETED_EDGES_CLAIM}, there exists a minimum spanning tree $T'$
of $G \setminus \{e\}$ so that $T' = T \cup \{e'\} \setminus \{e\}$ for an edge $e' \in E$. In
particular, $e' \in C_{T,e}$ is the minimum weight edge in this cut, and the partial cut
$\{e''\in C_{T,e}:\ w(e'') < w(e')\}$ is a candidate cut. Therefore $W \ge w(e')$ and
$p_{G}(e) =  w(e') - w(e) \le W - w(e)$.

Now, if $p_G(e) = 0$ then the assertion of the lemma is trivial. Otherwise, if $p_G(e) > 0$,
then $e$ must be contained in every minimum spanning tree of $G$. By the above
characterization of $p_G(e)$, we have that $p_G(e) = W - w(e)$, where $W$ is a minimum
weight of an edge in some cut $C_{G \setminus \{e\}}(S)$. As
$C' = C_{G \setminus B \setminus \{e\}}(S)\subset C_{G \setminus \{e\}}(S)$,
we have that $W' = \min\{w(e'):\ e'\in C'\} \ge W$. Using the same characterization
of $p_{G \setminus B}(e)$, we get $p_{G \setminus B}(e) \ge W' - w(e) \ge W - w(e) = p_G(e)$.
\end{proofof}

\begin{proofof}{Corollary~\ref{PROFIT_PRESERVATION}}
This is a simple application of Lemma~\ref{SUPERMODULARITY}, removing the
edges of $A = \{e_1, \ldots e_k\}$ one by one. Denote $A_0 = \emptyset$,
$A_1 = \{e_1\}$, $\dots$, $A_i =  \{e_1,e_2,\dots,e_i\}$, $\dots$, $A_k = \{e_1,e_2,\dots,e_k\}$.
By Lemma~\ref{SUPERMODULARITY}, for every $i=1,2,\dots,k$, it holds that
\[p_{G \setminus B \setminus A_{i-1}}(e_i) \ge p_{G \setminus A_{i-1}}(e_i).\]
Therefore,
\[p_{G \setminus B}(A) = \sum_{i=1}^{k} p_{G \setminus B \setminus A_{i-1}}(e_i) \ge
\sum_{i=1}^{k} p_{G \setminus A_{i-1}}(e_i) = p_{G}(A),\]
which completes the proof.
\end{proofof}

\begin{proofof}{Claim~\ref{cl: alg cost}} We show that the optimal solution $F^*$ is achieved at a partial cut considered by the
algorithm, and therefore the claim follows.

If $F^*$ disconnects $G$, then it is a global MIN CUT with respect to the edge costs
$c$. Moreover, all the edges in this cut have the same weight, otherwise removing just
the lightest edges would increase the weight of the minimum spanning tree, at lower
cost. Therefore, if the algorithm deals with one of the edges $e\in F^*$, one of the
feasible cuts it minimizes over is $F^*$.  Because $T$ is a spanning tree it must have
at least one edge in any complete cut in the graph, and specifically in $F^*$. In fact, in
this case the algorithm will output either $F^*$ or another global MIN CUT of the same
cost.

If $F^*$ does not disconnect $G$ we argue as follows. Let $T'$ be a minimum spanning
tree of $G \setminus F^*$. Recall that for every $e\in T \setminus T'$ there exists
$e'\in (T' \setminus T)\cap C_{T,e}$ such that $T-e+e'$ is a spanning tree. Let's consider
the edges $e\in T \setminus T'$ in arbitrary order, and let's choose $e' = \pi(e)$ that minimizes
$w(e')$ among all edges in $(T' \setminus T)\cap C_{T,e}$. Let $e_1$ denote the first edge
considered in $T \setminus T'$. As $T$ is a minimum spanning tree, $w(\pi(e_1))\ge w(e_1)$.
Denote $T_0=T$ and $T_1=T-e_1+\pi(e_1)$. If $w(\pi(e_1)) = w(e_1)$, we can repeat this
argument with $T_1$ and $T'$ to get $T_2$, and so forth. This process must reach an iteration
$i\le |T \setminus T'|$ at which $w(\pi(e_i)) > w(e_i)$, otherwise $w(T') = w(T)$, in contradiction
to the definition of $F^*$.

Now, consider the cut $C_{T_{i-1},e_i}$ in $G$. Notice that by construction, $T_{i-1}$ is also
a minimum spanning tree of $G$, because all exchanges prior to step $i$ did not increase the
weight of the tree. Thus, $w(e_i)$ is the minimum length of an edge in this cut. Also, $\pi(e_i)$
is an edge in this cut. By our choice of $\pi(e_i)$, none of the edges in the set
$F = \{e'\in C_{T_{i-1},e_i}:\ w(e') < w(\pi(e_i))\}$ are in $T'$. If there exists $e'\in F\setminus F^*$,
then the cycle closed by adding $e'$ to $T'$ must contain at least one other edge
$e''\in T'\cap C_{T_{i-1},e_i}$. However, all such edges have $w(e'') > w(e')$, in contradiction to
the assumption that $T'$ is a minimum spanning tree of $G \setminus F^*$. Thus, $F\subseteq F^*$.

Now, consider $F'\subseteq F$, putting $F' = \{e'\in C_{T_{i-1},e_i} | w(e') = w(e_i)\}$.
Let $T''$ be a minimum spanning tree of the graph $G \setminus F'$. Clearly, $w(T'') > w(T)$ and
$c(F')\le c(F^*)$. Hence, $F'$ is an optimal solution which contains all the minimum-weight edges
in the cut $C_{T_{i-1},e_i}$. Let's assume in contradiction that the minimum spanning tree $T$ that
the algorithm chooses and iterates over its edges maintains $T \cap F' = \emptyset$. Then there
exists an edge $e \in T$,  with $w(e) > w(e_i)$ that crosses the cut, and we can replace it and
create lighter spanning tree as $w(T - e + e_i) < w(T)$. This contradict $T$ being a minimum
spanning tree. We conclude that there is an edge $e \in T \cap F'$. Therefore, $F'$ is one of the
cuts that the algorithm optimizes over when considering $e$. The algorithm may choose a different
cut for $e$, but the chosen cut will not have cost greater than $c(F')$.
\end{proofof}

\begin{proofof}{Claim~\ref{cl: alg len}}
In this proof, we will use repeatedly the {\em blue rule}; see the proof of
Lemma~\ref{MST_DELETED_EDGES_CLAIM} for details.

Consider the best $e$ and the corresponding cut $C$ that determines the output of the
algorithm. There exists a minimum spanning tree $T$ of $G$ that contains $e$, because
we can apply the blue rule to $C$ and $e$. Let $S$ be the forest that remains of $T$ after
removing all the edges of length $w(e)$ in $C$. Clearly, $S$ has at least two connected
components, at least one on each side of the cut $C_{T,e}$ (be aware that this cut may
differ from $C$).

If the new graph is disconnected, then clearly the claim holds. Otherwise, $S$ can be
extended to a minimum spanning tree $T'$ of the new graph, as we can use the blue
rule to color blue each edge $f\in S$, using the cut $C_{T,f}$ that does not contain any
other edge of $T$ (Clearly $f$ has minimum length in this cut prior to the removal of edges,
and therefore also after the removal of edges). Let's assume for contradiction that
$w(T') = w(T)$.

Let $P$ be the cycle created by adding $e$ to $T'$. As $e$ crosses $C$, there must be
another edge $e\in P$ that crosses $C$. We must have that $w(e') > w(e)$, as we eliminated
from $C$ all the edges of length $w(e)$. By the assumption, it must be that $e'\in T$, otherwise
$w(T') > w(T)$ (exchanging $e'$ with $e$ reduces the cost of the tree, but $T$ is a minimum
spanning tree of $G$). Consider now $C_{T,e'}$. As $e'$ crosses $C_{T,e'}$, there must be
another $e''\in P$ that crosses $C_{T,e'}$. However, $e''\not\in T$, because
$C_{T,e'}\cap T = \{e'\}$ by definition. Thus, by our assumption $w(e'') = w(e) < w(e')$ (for the
same reason that, otherwise, exchanging $e''$ with $e$ reduces the length of $T'$, but we
assumed that $w(T') = w(T)$ and $T$ is a minimum spanning tree of $G$). This is a contradiction
to the assumption that $T$ is a minimum spanning tree, together with the implications that
$e'\in T$ and $e''\in C_{T,e'}\setminus\{e'\}$.
\end{proofof}

\bibliographystyle{plain}

\end{document}